\documentclass{article}
\usepackage[russian]{babel}
\usepackage[cp1251]{inputenc}
\usepackage[T2B]{fontenc}
\usepackage[pdftex, colorlinks, urlcolor=blue, pdfborder={0 0 0 [2 0]}]{hyperref}
\usepackage{amsfonts,amsmath,amssymb}
\usepackage{amsthm}
\usepackage{array}
\usepackage{indentfirst}

\newtheorem{theorem1}{\hspace{\parindent}\bf{Теорема}}
\newtheorem{lemma}{\hspace{\parindent}\bf{Лемма}}
\newtheorem{remark}{\hspace{\parindent}\bf{Замечание}}

\newtheorem{corol}{\hspace{\parindent}\bf{Следствие}}

\makeatletter
\renewcommand{\@biblabel}[1]{#1.} 
\makeatother

\begin{document}

\title{Нахождение решения характеристического уравнения \\
для модели диффузии  Ланжевена с ортогональными возмущениям \\
Constructing a solution to the characteristic equation for the Langevin diffusion model with orthogonal perturbations}

\author {В.А. Дубко\footnote{Научно-учебный центр прикладной информатики Национальной Аквдемии наук Украины, Киев, Украина}, С.В. Зубарев\footnote{Научно-учебный центр прикладной информатики Национальной Аквдемии наук Украины, Киев, Украина}, Е.В. Карачанская\footnote{Дальневосточный государственный университет путей сообщения, Хабаровск, Россия}}



\date{}

\maketitle
\begin{abstract}
Для модели Ланжевена динамики броуновской частицы с ортогональными к её текущей скорости возмущениями, в режиме, когда модуль скорости частицы становится постоянной, построено уравнение для характеристической функции $\psi (t,\lambda )=M\left[\exp (\lambda ,x(t))/V={\rm v}(0)\right]$ положения $x(t)$ броуновской частицы.
Полученные результаты подтверждают вывод о том, что модель динамики броуновской частицы,  построенная на основе нетрадиционной физической трактовки уравнений Ланжевена -- стохастических уравнений с ортогональными воздействиями, приводит к трактовке ансамбля броуновских частиц как системы, обладающей волновыми свойствами.
Эти результаты согласуются с ранее полученным выводам о том, при определённом согласовании коэффициентов в исходном стохастическом уравнении, для малых случайных влияниях и трении, уравнения Ланжевена приводят к описанию плотности вероятности положения частицы на основе волновых уравнений. При больших случайных воздействиях и трении плотность вероятности является решением диффузионного уравнения, с коэффициентом диффузии, меньшим по сравнению с моделью классической диффузии.

For the Langevin model of the dynamics of a Brownian particle with perturbations orthogonal to its current velocity, in a regime when the particle velocity modulus becomes constant, an equation for the characteristic function $\psi (t,\lambda )=M\left[\exp (\lambda ,x(t))/V={\rm v}(0)\right]$ of the position $
 x (t) $ of the Brownian particle.
The obtained results confirm the conclusion that the model of the dynamics of a Brownian particle, which constructed on the basis of an unconventional physical interpretation of the Langevin equations, i. e. stochastic equations with orthogonal influences, leads to the interpretation of an ensemble of Brownian particles as a system with wave properties.
These results are consistent with the previously obtained conclusions that, with a certain agreement of the coefficients in the original stochastic equation, for small random influences and friction, the Langevin equations lead to a description of the probability density of the position of a particle based on wave equations. For large random influences and friction, the probability density is a solution to the diffusion equation, with a diffusion coefficient that is lower than in the classical diffusion model.
\end{abstract}

\section*{Введение}
В работе \cite{18}  рассмотрено уравнение Ланжевена с ортогональными возмущениями к скорости броуновской частицы:
\begin{subequations}\label{GrindEQ__2_}
\begin{align}
 & \displaystyle d{\rm v}(t)=-a(t){\rm v}(t)dt+\frac{b(t)}{\left|{\rm v}(t)\right|} \left[{\rm v}(t)\times d{\rm w}(t)\right],
\label{GrindEQ__2a_} \\
&   \displaystyle dx(t)={\rm v}(t)dt, \ \ \  \ \ {\rm  v}(t), \ \ x(t)\in R^{3} ,
\label{GrindEQ__2b_}.
\end{align}
\end{subequations}
В  предположении  о независимости $x(0)$ и ${\rm v}(0)$, исследованы решения  уравнения для плотности распределения
\[\rho (t,x/{\rm v}(0))=\int _{-\infty }^{\infty } \rho (t,x/y;{\rm v}(0))\rho (y)dy_{1} dy_{2} dy_{3} . \]

При  постоянных значениях коэффициентов $a$ и $b$ это уравнение имеет вид:
\begin{equation} \label{GrindEQ__4_}
\begin{array}{c}
\displaystyle{\frac{\partial ^{2} \rho (t,x/{\rm v}(0))}{\partial t^{2} } +a\frac{\partial \rho (t,x/{\rm v}(0))}{\partial t} =} \\
+\displaystyle{3^{-1} |{\rm v}(0)|^{2} [1-\exp \{ -3at\} ]\nabla _{x}^{2} \rho (t,x/{\rm v}(0))} 
\displaystyle{+\sum _{l=1}^{3} \sum _{j=1}^{3} {\rm v}_{l} {\rm (}0)v_{j} (0)\frac{\partial ^{2} \rho (t,x/{\rm v}(0))}{\partial x_{l} \partial x_{j} } } \end{array}
\end{equation}

Были исследованы свойства решений уравнения \eqref{GrindEQ__4_} при различных соотношениях между коэффициентами:
\[a=\tilde{a}\varepsilon ^{-1} ,\; \; b=\tilde{b}\varepsilon ^{-1} ,\; \; \varepsilon \to 0, \]
\[a=\varepsilon a_{0} ,\; \; \displaystyle\frac{b^{2} }{\varepsilon a_{0}} =|{\rm v}(0))|^{2} =const, \ \  \varepsilon \to 0, \]
 где $\tilde{a},  \, \, a_{0} , \, \, \tilde{b}$ -- ограниченные величины.

В первом случае (диффузионное приближение) наблюдается замедление процесса диффузии по направлению $\displaystyle\frac{{\rm v}(0)}{{\rm |v}(0)|} $. Во втором (случай слабых взаимодействий со средой) -- плотность вероятности положения частицы аппроксимируется решением волнового уравнения.

Именно эти особенности асимптотических свойств решений \eqref{GrindEQ__4_}  стали стимулом к исследованию свойств ансамбля броуновских частиц, подчинённых уравнениям \eqref{GrindEQ__4_}, для произвольных времен при значениях коэффициентов, отличных от асимптотических, рассмотренных в \cite{18}.

Отметим, что данная модель не исчерпывает множества вариантов   интерпретации и применения моделей с ортогональными случайными воздействиями \cite{19,20}.

\section{Нахождение решения уравнения для характеристической функции процесса $x(t)$}

Будем искать решение уравнение для характеристической функции уравнения \eqref{GrindEQ__4_} когда $a$ и $b$  -- постоянные величины, и $\displaystyle\frac{b^{2}}{a} =|{\rm v}(0))|^{2} =const$. Эти условия соответствуют тому, что начальное значение  ${\rm v}(0)=V$ находится на поверхности устойчивости для процесса ${\rm v}(t;{\rm  v(0)})$.

В этом случае уравнение для характеристической функции \eqref{GrindEQ__4_}
при условии, что $x(0),{\rm v}(0)$  --  независимые, имеет вид:
\begin{subequations}\label{GrindEQ__5_}
\begin{align}
 & \displaystyle\frac{\partial ^{2} \psi (t)}{\partial t^{2} } +a\frac{\partial \psi (t)}{\partial t} =-|\lambda |^{2} \psi (t)\left[\frac{\left|V\right|^{2} }{3} (1-e^{-3at} )\right]-(\lambda ,{\rm v}(0))^{2} \psi (t)
\label{GrindEQ__5a_} \\
& \displaystyle\psi (0)=M[\exp \{ {\rm i}\left(\lambda ,x{\rm (}0)\right)\} ],\; \, \, \, \, \; \frac{\partial   \psi (t)}{\partial t  } \left|_{t=0} \right. ={\rm i}(\lambda ,{\rm v}(0))\psi (0),\; \, \, \, \; {\rm i}^{2} =-1,
\label{GrindEQ__5b_}
\end{align}
\end{subequations}

Перепишем, первоначально, уравнение \eqref{GrindEQ__5a_} в таком виде:
\begin{equation} \label{GrindEQ__6_}
\frac{d^{2} \psi (t)}{dt^{2} } +a\frac{d\psi (t)}{dt} +(-\alpha ^{2} \cdot e^{-3at} +\beta ^{2} )\cdot \psi (t)=0,
\end{equation}
 где
\begin{equation} \label{GrindEQ__7_}
\; \alpha =\frac{\left|\lambda \right|\left|V\right|}{\sqrt{3} } , \ \ \beta ^{2} =\alpha ^{2} +(\lambda ,V)^{2}
\end{equation}

 Выполним  замены: $\psi (t)=u(\tau )$ и $\tau =-3at$. Теперь уравнение \eqref{GrindEQ__6_} принимает вид:
\begin{equation} \label{GrindEQ__8_}
\displaystyle\frac{d^{2} u(\tau )}{d\tau ^{2} } -\frac{1}{3} \frac{du(\tau )}{d\tau } +\left[-\left(\frac{\alpha }{3a} \right)^{2} \cdot e^{\tau } +\left(\frac{\beta }{3a} \right)^{2} \right]\cdot u(\tau )=0.
\end{equation}

После выполнения замены $\displaystyle z(\tau )=\left(\frac{2\alpha }{3a} \right)e^{\frac{\tau }{2} } \ge 0$, $u(\tau )=e^{\frac{\tau }{6} } W(z(\tau ))$, приходим к уравнению:
\begin{equation} \label{GrindEQ__9_}
z^{2} \frac{d^{2} W(z)}{dz^{2} } +z\frac{dW(z)}{dz} -\left[z^{2} +\upsilon ^{2} \right]\cdot W(z)=0
\end{equation}
где
\begin{equation} \label{GrindEQ__10_}
\upsilon =\frac{1}{3a} \sqrt{a^{2} -4\beta ^{2} }.
\end{equation}

Последнее уравнение -- модифицированное уравнение Бесселя. Его решениями  являются модифицированными функциями Бесселя 1-го рода \cite[с.13]{21}:
\begin{equation}\label{GrindEQ__11_}
 W_{1} (z)=I_{+\upsilon } (z)  \ \ \ \ {\mbox{\rm и}} \ \ \ \ W_{2} (z)=I_{-\upsilon } (z),
\end{equation}
где
\begin{equation} \label{GrindEQ__12_}
I_{\pm \upsilon } (z)=\sum _{m=0}^{\infty } \frac{\left(z2^{-1} \right)^{2m\pm \upsilon } }{m!\Gamma \left(m\pm \upsilon +1\right)}
\end{equation}

 Необходимо рассмотреть два варианта:

\textbf{Вариант 1}.    $a^{2} -4\beta ^{2} \ge 0$, т. е. $\upsilon $ -- действительное число.

\textbf{Вариант 2.}     $a^{2} -4\beta ^{2} <0$, т. е. $a<2\beta $. Тогда $\upsilon =\pm {\rm i}\gamma $, где
$\displaystyle\gamma =\frac{1}{3a} \sqrt{4\beta ^{2} -a^{2} } $  -- действительное число, которое может быть, как целым так и не целым.

\begin{lemma}\label{L1}
При выполнении условия $a^{2} -4\beta ^{2} \ge 0$ параметр $\lambda$ удовлетворяет ограничению:
\begin{equation} \label{GrindEQ__13_}
|\lambda |\in \left[0,a\frac{2\sqrt{3} }{\left|V\right|} \right].
\end{equation}
\end{lemma}

\begin{proof}
С учётом обозначений \eqref{GrindEQ__7_} имеем:
$$
\alpha ^{2} =\displaystyle\frac{\left|\lambda \right|^{2} \left|V\right|^{2} }{3}, \ \ \ \ \beta ^{2} =\displaystyle\frac{\left|\lambda \right|^{2} \left|V\right|^{2} }{3} +(\lambda ,V)^{2} .
$$
 Следовательно,
 $\displaystyle \upsilon =\frac{1}{3a} \sqrt{a^{2} -4\beta ^{2} } <\frac{1}{3} $,
 т.к.
 $$\displaystyle\mathop{\inf }\limits_{\lambda } \beta ^{2} (\lambda )=\left. \left(\frac{\left|\lambda \right|^{2} \left|V\right|^{2} }{3} +(\lambda ,V)^{2} \right)\right|_{\lambda =0} =0.$$

Условие положительности $a^{2} -4\beta ^{2} $ приводит  к требованию
\[1-\frac{\left|\lambda \right|^{2} \left|V\right|^{2} }{4a^{2} } \left(\frac{1}{3} +\frac{(\lambda ,V)^{2} }{\left|\lambda \right|^{2} \left|V\right|^{2} } \right)\ge 0\]
 или
\[1-\frac{1}{4a^{2} } \left(\left|\lambda \right|^{2} \left|V\right|^{2} 3^{-1} +(\lambda ,V)^{2} \right)\ge 0.\]

Таким образом,  точка $|\lambda |=0$ обеспечивает последнее  неравенство.
С другой стороны, это неравенство можно представить в таком виде:
\[\left|\lambda \right|^{2} \le \frac{4a^{2} }{\left|V\right|^{2} \left(\frac{1}{3} +\frac{(\lambda ,V)^{2} }{\left|\lambda \right|^{2} \left|V\right|^{2} } \right)}. \]

Далее,
\[\left|\lambda \right|^{2} \le \frac{4a^{2} }{\left|V\right|^{2} \left(\frac{1}{3} +\frac{(\lambda ,V)^{2} }{\left|\lambda \right|^{2} \left|V\right|^{2} } \right)} \le \mathop{\sup }\limits_{\lambda } \frac{4a^{2} }{\left|V\right|^{2} \left(\frac{1}{3} +\frac{(\lambda ,V)^{2} }{\left|\lambda \right|^{2} \left|V\right|^{2} } \right)} =\frac{12a^{2} }{\left|V\right|^{2} }   , \]
т.к.  знаменатель достигает минимального значения при $(\lambda ,V)=0$.

Эти неравенства и доказывают справедливость \eqref{GrindEQ__13_}.
\end{proof}

\begin{corol}\label{S1}
 Неравенство $a^{2} -4\beta ^{2} <0$ выполняется для любых $|\lambda |\notin \left[0,a\displaystyle\frac{2\sqrt{3} }{\left|V\right|} \right]$.
\end{corol}

Для нецелых значений параметра $\upsilon $  функции  $I_{+\upsilon } (z)$ и  $I_{-\upsilon } (z)$  линейно независимы. Следовательно, в области \eqref{GrindEQ__13_} общее решение уравнения \eqref{GrindEQ__5a_}, с учётом представления \eqref{GrindEQ__12_}, при возвращении к переменной $t$  имеет вид:
\begin{equation} \label{GrindEQ__14_}
\psi (t)=e^{-at/2} [C_{1} I_{+\upsilon } (2p\eta (t))+C_{2} I_{-\upsilon } (2p\eta (t)],
\end{equation}
где $C_{1} ,C_{2} $ -- постоянные, и, для удобства восприятия, введены обозначения:
\begin{equation} \label{GrindEQ__15_}
\displaystyle\eta (t)=e^{-\frac{3at}{2} } ,\; \; p=\frac{\alpha }{3a} \ge 0.
\end{equation}

\begin{remark}\label{Zm1}
Учитывая обозначения \eqref{GrindEQ__7_}, получаем: $\displaystyle p=\frac{\alpha }{3a}   =\frac{\left|\lambda \right|\left|V\right|}{3a\sqrt{3} } $ и, следовательно, для любого $\lambda $ в области изменения \eqref{GrindEQ__13_} следует, что $\displaystyle  p\in [0,{\raise0.5ex\hbox{$\scriptstyle 2 $}\kern-0.1em/\kern-0.15em\lower0.25ex\hbox{$\scriptstyle 3 $}}]$.
\end{remark}

\begin{lemma}\label{L2}
Общее решение \eqref{GrindEQ__6_} для нецелых $\upsilon $ допускает представление:
\begin{equation} \label{GrindEQ__16_}
\psi (t)=C_{1} F_{+\upsilon } (t)e^{-\, \, \frac{[1+3\upsilon ]at}{2} } +C_{2} F_{-\upsilon } (t)e^{-\, \, \frac{[1-3\upsilon ]at}{2} } ,
\end{equation}
где
\begin{equation} \label{GrindEQ__17_}
F_{\pm \upsilon } (t)=\sum _{m=0}^{\infty } \frac{\left[p\cdot \eta (t)\right]^{2m} }{\Gamma (m\pm \upsilon +1)m!}
\end{equation}
\end{lemma}

\begin{proof}
Рассмотрим явное представление $I_{\pm \upsilon } (p\eta (t))$:
\begin{equation} \label{GrindEQ__18_}
I_{\pm \upsilon } (2p\eta (t))=\sum _{m=0}^{\infty } \frac{\left(pe^{-\frac{3at}{2} } \right)^{2m\pm \upsilon } (i)^{2m\pm \upsilon } }{m!\Gamma \left(m\pm \upsilon +1\right)} =-F_{\pm \upsilon } (pe^{-\, \, \frac{3at}{2} } )({\rm i}p)^{\pm \upsilon } e^{\mp \frac{3\upsilon at}{2} } .
\end{equation}

Для того, что бы включить значение $p=0$, что соответствует, по построению, $\displaystyle|\lambda |=a\frac{2\sqrt{3} }{\left|V\right|} $, будем учитывать тот факт,  что если $W(z)$ -- частное решение \eqref{GrindEQ__9_}, то и $\mu W(z)$, где $\mu $ -- произвольная постоянная, также является частным решением уравнения \eqref{GrindEQ__9_}.

Это позволяет, с учётом представления \eqref{GrindEQ__12_} и постоянства значения выражения $({\rm i}p)^{\pm \upsilon }$ при любых значениях параметра $\upsilon $, искать частные решения \eqref{GrindEQ__6_} в таком виде:
\begin{equation} \label{GrindEQ__19_}
\displaystyle\psi _{1} (t)=F_{+\upsilon } (t)e^{-\, \, \frac{[1+3\upsilon ]at}{2} } ,\ \ \ \psi _{2} (t)=F_{-\upsilon } (t)e^{-\, \, \frac{[1-3\upsilon ]at}{2} } .
\end{equation}
Соответственно, общее решение преобразуется в \eqref{GrindEQ__16_}, где постоянные \textit{С}${}_{1}$, \textit{С}${}_{2}$, определяются с использованием начальных условия \eqref{GrindEQ__5b_}.
\end{proof}

\begin{remark}\label{Zm2}
С учётом \eqref{GrindEQ__17_} и теоремы Лагранжа, функции $F_{\pm \upsilon } (t)$ являются всюду убывающим функциями по $t$ для всех значений $\upsilon \in [0,3^{-1} ]$. В силу Замечания \ref{Zm1} они ограничена в области \eqref{GrindEQ__13_} значений $\lambda $.
\end{remark}

Рассмотрим свойства функций $\psi _{\upsilon } (t)$ в представлении \eqref{GrindEQ__16_}.

\begin{lemma}\label{L3}
Для любых $\upsilon \in [0,3^{-1} ]$ функции  $\psi _{\upsilon } (t)$ монотонно сходятся к 0 при $t\to \infty $.
\end{lemma}

\begin{proof}
Достаточно установить справедливость этого утверждения отдельно для представлений  $\psi _{1} (t)$ и $\psi _{2} (t)$ \eqref{GrindEQ__19_}.  Докажем это для функции $\psi _{2} (t)$:
\[\displaystyle\psi _{2} (t)=e^{-\, \frac{t}{2} \left[a-\sqrt{(a^{2} -4\beta ^{2} )} \right]} \cdot F_{-\upsilon } (t).\]

Так как выражение $a-\sqrt{(a^{2} -4\beta ^{2} )} \ge 0$ при $a\ge 0$, a $F_{-\upsilon } (t)$ ограничено, в силу Замечания \ref{Zm2},   то $\mathop{\lim }\limits_{t\to \infty } \psi _{2} (t)=0$.
\end{proof}

  Аналогично устанавливается, с учётом Замечания \ref{Zm2}, что и $\mathop{\lim }\limits_{t\to \infty } \psi _{1} (t)=0$.

Обобщим полученные выводы в форме утверждения.

\begin{theorem1}\label{T2}
Общее решение $\psi (t)$ уравнения \eqref{GrindEQ__5_}  при условии
\[|\lambda |\in \left[0,a\frac{2\sqrt{3} }{\left|V\right|} \right],\]
имеет вид:
\begin{equation} \label{GrindEQ__20_}
\psi (t)=C_{1} F_{\upsilon } (t)e^{-\, \, \frac{[1+3\upsilon ]at}{2} } +C_{2} F_{-\upsilon } (t)e^{-\, \, \frac{[1-3\upsilon ]at}{2} } ,
\end{equation}
где \textit{С}${}_{1}$ и \textit{С}${}_{2}$ -- решения системы уравнений:
\begin{equation} \label{GrindEQ__21_}
\left\{\begin{array}{l}
{C_{1} F_{\upsilon } (0)+C_{2} F_{-\upsilon } (0){\rm \; }={\rm M}[\exp \left\{{\rm i}\left(\lambda ,x{\rm (}0)\right)\right\}]} \\
\\
\displaystyle {C_{1} \frac{(\partial e^{-\frac{[1-3\upsilon ]at}{2} } F_{\upsilon } (t))}{\partial t} \left|_{t=0} \right. +C_{2} \frac{\partial (e^{-\frac{[3\upsilon +1]at}{2} } F_{-\upsilon } (t))}{\partial t} \left|_{t=0} \right. =-{\rm i}(\lambda ,{\rm v}(0)){\rm M}[\exp \left\{{\rm i}\left(\lambda ,x{\rm (}0)\right)\right\}]}
\end{array}\right.
\end{equation}
\end{theorem1}

\begin{proof}
есть следствие Лемм \ref{L1}, \ref{L2}, \ref{L3}, Замечаний и подстановки общего решения \eqref{GrindEQ__20_} в начальные условия \eqref{GrindEQ__5b_}. Воспользовавшись явным видом \eqref{GrindEQ__17_} для $F_{\pm \upsilon } (t)$, их линейной независимостью, убеждаемся, что для любых $\upsilon \in [0,3^{-1} ]$ значения
\[\displaystyle F_{\upsilon } (0),\, \, \, \, \, \, F_{-\upsilon } (0),\ \ \ \ \frac{\partial (e^{-\frac{[1-3\upsilon ]at}{2} } F_{\upsilon } (t))}{\partial t} \left|_{t=0} \right. ,\, \, \, \, \, \, \, \frac{\partial e^{-\frac{[3\upsilon +1]at}{2} } F_{-\upsilon } (t)}{\partial t} \left|_{t=0} \right. \]
существуют, ограничены и не обращаются одновременно в нуль. Это указывает на существование  решения \eqref{GrindEQ__21_}.
\end{proof}

Перейдем к рассмотрению \textbf{Варианта 2}, т. е., когда
\[\displaystyle\upsilon =\pm {\rm i}\gamma , \ \   {\rm i}^{2} =-1, \ \  \gamma =\frac{1}{3a} \cdot \sqrt{4\beta ^{2} -a^{2} } .    \]

Запишем модифицированные функции Бесселя 1-го рода с чисто мнимым индексом в исходной форме  \cite[c.13]{17}. В нашем случае, с учётом обозначений \eqref{GrindEQ__15_} и представления \eqref{GrindEQ__17_}, они имеют вид:
\begin{equation} \label{GrindEQ__22_}
I_{\pm {\rm i}\gamma } \left[2p\cdot \eta (t)\right]=\left[p\cdot \eta (t)\right]^{\pm {\rm i}\gamma } \cdot F_{\pm {\rm i}\gamma } (t),
\end{equation}

\begin{equation} \label{GrindEQ__23_}
F_{\pm {\rm i}\gamma } (t)=\sum _{m=0}^{\infty } \frac{\left[p\cdot \eta (t)\right]^{2m} }{m!\Gamma (m+1\pm {\rm i}\gamma )}
\end{equation}

 Следовательно, комплексно-значные частные решения уравнения \eqref{GrindEQ__5a_}, имеют вид:
\begin{equation} \label{GrindEQ__24_}
\psi _{1,2} (t,{\rm i})=e^{-\frac{at}{2} } \cdot I_{\pm {\rm i}
\gamma } \left[p\cdot \eta (t)\right].
\end{equation}

\begin{lemma}\label{L4}
\begin{equation} \label{GrindEQ__25_}
F_{\pm {\rm i}\gamma } (t)=\Phi _{1} (t,\gamma )\mp {\rm i}\cdot \Phi _{2} (t,\gamma ),
\end{equation}
 где
\begin{equation} \label{GrindEQ__26_}
\displaystyle\Phi _{1} \left(t,\gamma \right)=\sum _{m=0}^{\infty } \frac{\left[p\cdot \eta (t)\right]^{2m} \cdot Q_{1} \left(\gamma ,m\right)}{m!\cdot \left[Q_{1}^{2} \left(\gamma ,m\right)+Q_{{\rm 2}}^{{\rm 2}} \left(\gamma ,m\right)\right]} ,
\end{equation}

\begin{equation} \label{GrindEQ__27_}
\displaystyle\Phi _{2} (t,\gamma )=\sum _{m=0}^{\infty } \frac{\left[p\cdot \eta {\rm (}t)\right]^{2m} \cdot Q_{2} \left(\gamma ,m\right)}{m!\cdot \left[Q_{1}^{2} \left(\gamma ,m\right)+Q_{2}^{2} \left(\gamma ,m\right)\right]},
 \end{equation}

\begin{equation} \label{GrindEQ__28_}
\displaystyle Q_{1} (\gamma ,m)=\int _{0}^{\infty } \cos \left(\gamma \ln \tau \right)\cdot \tau ^{m} \cdot e^{-\tau } d\tau ,                                        \end{equation}

\begin{equation} \label{GrindEQ__29_}
\displaystyle Q_{2} (\gamma ,m)=\int _{0}^{\infty } \sin \left(\gamma \ln \tau \right)\cdot \tau ^{m} \cdot e^{-\tau } d\tau .
\end{equation}

\end{lemma}

\begin{proof}
Выполним преобразование:
\begin{equation} \label{GrindEQ__30_}
\begin{array}{l}
\displaystyle{\left[p\cdot \eta (t)\right]^{\pm {\rm i}\gamma } =\left[p\cdot e^{-\frac{3at}{2} } \right]^{\pm {\rm i}\gamma } =p^{\pm {\rm i}\gamma } \cdot e^{\mp \frac{3a\gamma t}{2} } =e^{\ln \left(p\right)^{\pm {\rm i}\gamma } } \cdot e^{\mp {\rm i}\frac{3a\gamma t}{2} } =} \\
\displaystyle{=e^{\pm {\rm i}\left(\gamma \ln p-\frac{3a\gamma t}{2} \right)} =\cos \left(\gamma \ln p-\frac{3a\gamma }{2} t\right)\pm {\rm i}\cdot \sin \left(\gamma \ln p-\frac{3a\gamma }{2} t\right).}
\end{array}
\end{equation}
\begin{equation} \label{GrindEQ__31_}
\Gamma \left(m+1\pm {\rm i}\gamma \right)=\int _{0}^{\infty } \tau ^{{\rm m}\pm {\rm i}\gamma } \cdot e^{-\tau } d\tau =\; Q_{1} (\gamma ,m)\pm {\rm i}\cdot Q_{2} (\gamma ,m).
\end{equation}

Выделим в \eqref{GrindEQ__22_} и \eqref{GrindEQ__23_} действительные и мнимые части. В силу линейности уравнения \eqref{GrindEQ__5a_}, они также будут частными решениями \eqref{GrindEQ__5a_}.

Подставляя \eqref{GrindEQ__31_} в \eqref{GrindEQ__23_} и совершая стандартные операции с комплексными величинами, приходим к представлению \eqref{GrindEQ__25_}.

 Обращаем внимание на то, что равномерная сходимость \eqref{GrindEQ__26_} и \eqref{GrindEQ__27_} следует из равномерной сходимости $F_{\pm {\rm i}\gamma } (t)$ (см. \eqref{GrindEQ__23_}). Сходимость $F_{\pm {\rm i}\gamma } (t)$, в свою очередь, следует из равномерной сходимости интеграла Сонина -- Шлефли для любых комплексных индексов функции Бесселя  \cite[c.640--641]{22}.

Таким образом, решения уравнения \eqref{GrindEQ__5a_} с учётом \eqref{GrindEQ__30_}, \eqref{GrindEQ__24_} и \eqref{GrindEQ__25_} приобретут вид:
\begin{equation} \label{GrindEQ__32_}
\psi _{1,2} (t,{\rm i})=\left[\cos \left(\sigma -\mu t\right)\pm {\rm i}\sin \left(\sigma -\mu t\right)\right]\cdot \left[\Phi _{1} \left(t,\gamma \right)\mp {\rm i}\Phi _{2} \left(t,\gamma \right)\right]e^{-\, \, \frac{at}{2} } ,
\end{equation}
 где $\sigma =\gamma \ln p$, $\mu =3a\gamma $.

Раскрывая в \eqref{GrindEQ__32_} скобки, запишем $\psi _{1,2} (t,{\rm i})$ в таком виде:
\[\psi _{1} (t,{\rm i})=e^{-\frac{at}{2} } \left[\varphi _{1} \left(t\right)+{\rm i}\cdot \varphi _{2} \left(t\right)\right],\psi _{2} (t,{\rm i})=e^{-\frac{at}{2} } \cdot \left[\varphi _{1} (t)-{\rm i}\cdot \varphi _{2} (t)\right],\]
где
\begin{equation} \label{GrindEQ__33_}
\varphi _{1} (t)=\Phi _{1} (t,\gamma )\cdot \cos \left(\sigma -\mu t\right)+\Phi _{2} (t,\gamma )\cdot \sin \left(\sigma -\mu t\right),
\end{equation}

\begin{equation} \label{GrindEQ__34_}
\varphi _{2} (t)=\Phi _{1} (t,\gamma )\cdot \sin \left(\sigma -\mu t\right)-\Phi _{2} (t,\gamma )\cdot \cos \left(\sigma -\mu t\right).
\end{equation}
\end{proof}

\begin{lemma}\label{L5}
На основе  функций $\psi _{1} (t)=e^{-\frac{at}{2} } \cdot \varphi _{1} (t)$ и $\psi _{2} (t)=e^{-\frac{at}{2} } \cdot \varphi _{2} (t)$ возможно построить фундаментальную систему решений уравнения \eqref{GrindEQ__5a_}.
\end{lemma}

\begin{proof}
Для того, чтобы убедиться в этом, установим линейную независимость частных решений $\displaystyle \psi _{1} (t)=e^{-\frac{at}{2} } \cdot \varphi _{1} (t)$ и $\psi _{2} (t)=e^{-\frac{at}{2} } \cdot \varphi _{2} (t)$. Для этого необходимо установить, что соответствующий определитель Вронского не равен нулю. В нашем случае он имеет вид:
\[\displaystyle W\left[\psi _{1} (t),\psi _{2} (t)\right]=
\det \left(
\begin{array}{cc} {\psi _{1} } & {\psi _{2} } \\ {\frac{d\psi _{1} }{dt} } & {\frac{d\psi _{2} }{dt} }
\end{array}
\right)=
\det \left(
\begin{array}{cc} {e^{-\frac{at}{2} } \cdot \varphi _{1} } & {e^{-\frac{at}{2} } \cdot \varphi _{2} } \\ {\frac{d\left[e^{-\frac{at}{2} } \cdot \varphi _{1} \right]}{dt} } & {\frac{d\left[e^{-\frac{at}{2} } \cdot \varphi _{2} \right]}{dt} }
\end{array}
\right).\]

Дифференцируя элементы второй строки, приходим к следующему представлению $W(\psi _{1} ,\psi _{2} )$:
\[W(\psi _{1} ,\psi _{2} )=e^{-\frac{at}{2} } \cdot \det \left(\begin{array}{cc} {\varphi _{1} } & {\varphi _{2} } \\ {\frac{d\varphi _{1} }{dt} } & {\frac{d\varphi _{2} }{dt} } \end{array}\right)=e^{-\frac{at}{2} } \cdot W\left(\varphi _{1} ,\varphi _{2} \right).\]

Таким образом:
\begin{equation} \label{GrindEQ__35_}
W(\psi _{1} ,\psi _{2} )=e^{-\frac{at}{2} } \cdot W(\varphi _{1} ,\varphi _{2} ).
\end{equation}

 Подставим \eqref{GrindEQ__33_} и \eqref{GrindEQ__34_} в \eqref{GrindEQ__35_}. После преобразований приходим к следующему виду для $W(\varphi _{1} ,\varphi _{2} )$:
\begin{equation} \label{GrindEQ__36_}
W(\varphi _{1,} \varphi _{2} )=\left(\Phi _{2} \right)^{2} \cdot \frac{d\left[\frac{\Phi _{1} }{\Phi _{2} } \right]}{dt} -\mu \cdot \left[\left(\Phi _{1} \right)^{2} +\left(\Phi _{2} \right)^{2} \right].
\end{equation}

Покажем, что \eqref{GrindEQ__36_}, применительно для функций $\Phi _{1} $ и $\Phi _{2} $, не обращается в 0.

Допустим, что это не так. Тогда \eqref{GrindEQ__36_} будет дифференциальным уравнением 1-го порядка относительно $y=\left(\frac{\Phi _{1} }{\Phi _{2} } \right)$:
\[\frac{dy}{dt} =\mu \left[1+\left(y\right)^{2} \right].\]

Решая данное уравнение, получаем следующее соотношение для функций $\Phi _{1} $ и $\Phi _{2} $:
\begin{equation} \label{GrindEQ__37_}
\Phi _{1} (t)=\Phi _{2} (t)\cdot {\rm tg}\left(\mu t+C\right),
\end{equation}
где $C$ -- постоянная интегрирования.

Например, в точках $t$, удовлетворяющих равенству $\mu t+C=\pi k$,  $k=0,\pm 1,\pm 2...$, выражение \eqref{GrindEQ__26_} не обращается в 0 при всех конечных значениях постоянной $C$. Таким образом, пришли к противоречию, поскольку  \eqref{GrindEQ__26_} функции $\Phi _{1} $ не согласуется с  \eqref{GrindEQ__37_}.

Следовательно, при чисто мнимом индексе для функций Бесселя фундаментальная система решений уравнения \eqref{GrindEQ__5a_} имеет вид:
\[\psi (t)=e^{-\, \, \frac{at}{2} } \left[C_{1} \varphi _{1} (t)+C_{2} \varphi _{2} (t)\right],\]
где $C_{1} $, $C_{2} $ -- постоянные, определяемые через начальные условия, а $\varphi _{1} (t)$ и $\varphi _{2} (t)$ определяются равенствами \eqref{GrindEQ__33_} и \eqref{GrindEQ__34_}.
\end{proof}

\begin{theorem1}\label{T3}
Общее решение  $\psi (t)$ уравнения \eqref{GrindEQ__5_} при условии $|\lambda |\notin \left[0,a\displaystyle\frac{2\sqrt{3} }{\left|V\right|} \right]$,
существует и имеет вид:
\begin{equation} \label{GrindEQ__38_}
\psi (t)=e^{-\frac{at}{2} } \cdot \left[C_{1} \varphi _{1} (t)+C_{2} \varphi _{2} (t)\right],
\end{equation}
\[\varphi _{1} (t)=\Phi _{1} (t,\gamma )\cdot \cos \left(\sigma -\mu t\right)+\Phi _{2} (t,\gamma )\cdot \sin \left(\sigma -\mu t\right),                              \]
\[\varphi _{2} (t)=\Phi _{1} (t,\gamma )\cdot \sin \left(\sigma -\mu t\right)-\Phi _{2} (t,\gamma )\cdot \cos \left(\sigma -\mu t\right).\]
где ${С}_{1}$ и ${С}_{2}$ -- решения системы уравнений:
\begin{equation} \label{GrindEQ__39_}
\left\{\begin{array}{l}
{\left[C_{1} \varphi _{1} (0)+C_{2} \varphi _{2} (0)\right],{\rm \; }={\rm M}[\exp \left\{{\rm i}\left(\lambda ,x{\rm (}0)\right)\right\}/V]} \\
\displaystyle{C_{1} \frac{\partial \varphi _{1} (t)}{\partial t} \left|_{t=0} \right. +C_{2} \frac{\partial \varphi _{2} (t)}{\partial t} \left|_{t=0} \right. =-[{\rm i}(\lambda ,{\rm v}(0))-{\rm a2}^{{\rm -1}} {\rm ]M}[\exp \left\{{\rm i}\left(\lambda ,x{\rm (}0)\right)\right\}/V]}
\end{array}\right.
\end{equation}
где $\sigma =\gamma \ln p$, $\mu =3a\gamma $, $\gamma =\displaystyle\frac{1}{3a} \cdot \sqrt{4\beta ^{2} -a^{2} } $, $p=\displaystyle\frac{|\lambda ||V|}{3a\sqrt{3} } $, а  функции $\Phi _{1} (t,\gamma )$, $\Phi _{2} (t,\gamma )$ определяются выражениями \eqref{GrindEQ__26_}, \eqref{GrindEQ__27_}.
\end{theorem1}

 \begin{proof}
Представление \eqref{GrindEQ__38_} есть  следствие утверждений Лемм \ref{L4}, \ref{L5}, Замечания \ref{Zm2} и взаимосвязи между переменными для Варианта 2.  Условия \eqref{GrindEQ__39_} -- результат подстановки значений функций  $\varphi _{1,2} (t)$ и их производных в начальные условия \eqref{GrindEQ__5b_}.
 \end{proof}

\begin{remark}
Представление \eqref{GrindEQ__38_} будет справедливо как для целых так и не целых значений~$\gamma $.
\end{remark}

\section*{Выводы}
Таким образом, полученные результаты
(теоремы \ref{T2} и \ref{T3}) в спектре характеристической функции, являющейся решением \eqref{GrindEQ__5_}, присутствует область  колебательных составляющих для любых значений параметров $a$ и $b$, при условии $\displaystyle\frac{b^{2}}{a} =|{\rm v}(0)|^{2} =const$.

Следовательно, полученные результаты подтверждают вывод о том, что для эволюция ансамбля броуновских частиц, траектории которых являются решением  уравнений типа Ланжевена \cite{23}, но при нетрадиционной физической трактовке влияния случайных сил (стохастические уравнения с ортогональными воздействиями) связана с колебательными процессами. Эти результаты согласуются с выводами работы \cite{18}, в которой показано, что при определённом согласовании коэффициентов в исходном стохастическом уравнении, в случае малых случайных воздействиях и трении, эволюция плотности вероятности для ансамбля частиц есть решение волнового уравнения. При большом трении и больших случайных воздействиях плотность вероятности положения частиц для этой модели, как показано в  \cite{5,7,18}, является решением диффузионного уравнения с коэффициентом диффузии, меньшим по сравнению с получаемым в моделях классической диффузии. Таким образом, для данной модели динамики броуновской частицы, ансамбль  броуновских частиц  --  это поток, обладающий волновыми свойствами.

\end{document}